\documentclass[conference,letterpaper]{IEEEtran}

\usepackage[utf8]{inputenc} 
\usepackage[T1]{fontenc}
\usepackage{url}
\usepackage{ifthen}
\usepackage{cite}
\usepackage[cmex10]{amsmath} 
\usepackage{algorithmic}
\usepackage{algorithm}
\usepackage{graphicx}
\usepackage{cite}
\usepackage{array}
\usepackage{amsmath,amsfonts}
\usepackage{amssymb}
\usepackage{microtype}
\usepackage{comment}
\usepackage{bm}
\newtheorem{theorem}{Theorem}
\newtheorem{definition}{Definition}   
\newtheorem{lemma}{Lemma}      
\newtheorem{corollary}{Corollary}    
\newtheorem{example}{Example}       
\newtheorem{remark}{Remark}

\def\BibTeX{{\rm B\kern-.05em{\sc i\kern-.025em b}\kern-.08em
    T\kern-.1667em\lower.7ex\hbox{E}\kern-.125emX}}

\IEEEoverridecommandlockouts
\begin{document}
\bstctlcite{BSTcontrol}

\title{ Existence and Counting Bounds for High-Memory Spatially-Coupled Codes via the Combinatorial Nullstellensatz 
\thanks{This work is partially supported by the National Key R\&D Program of China, (2023YFA1009600).}}

 \author{
 \IEEEauthorblockN{Lei Huang}
 \IEEEauthorblockA{Data Science Institute\\
                    Shandong University\\
                    Jinan, China\\
                    Email: leihuang@mail.sdu.edu.cn}}

\maketitle

\begin{abstract}

The finite-length performance of spatially-coupled low-density parity-check (SC-LDPC) codes is strongly affected by short cycle configurations and the harmful structures induced by them. 
This paper studies SC-LDPC code design directly at the protograph level, where the design variables are the edge-spreading assignments specified by the partition matrix. 
In contrast to CLLL/Moser--Tardos based constructive frameworks for QC-SC-LDPC codes, we focus on sharper nonconstructive existence and counting bounds. 
By encoding cycle-activation conditions as polynomial vanishing constraints over finite grids, we apply the Combinatorial Nullstellensatz to derive sufficient memory conditions for eliminating prescribed cycle-induced harmful structures. 
For fully connected $(\gamma,\kappa)$ base graphs, the resulting bounds explicitly characterize the memory required to destroy all $4$-cycles as well as all $4$- and $6$-cycles, and for fixed $\gamma$, they are asymptotically tight up to a constant factor compared with known lower bounds.
We further apply the Alon--Füredi theorem to obtain lower bounds on the number of feasible edge-spreading assignments, including an explicit counting bound for assignments that eliminate all $4$-cycles and hence yield girth at least six. 
These results provide a refined algebraic-combinatorial characterization of the feasible design space for high-memory SC-LDPC codes, although no corresponding construction algorithm is provided.
\end{abstract}

\section{Introduction}

The threshold-saturation effect of spatially-coupled low-density parity-check (SC-LDPC) codes \cite{r1,r2,r3}, together with their capacity-achieving behavior over general binary memoryless symmetric channels \cite{r4}, has established SC-LDPC codes as a powerful family of modern channel codes \cite{r5,r6,r7,r8,r9,r10,r11,r12}. 
Their convolutional structure further enables windowed decoding (WD), which provides low-latency iterative decoding with controllable complexity \cite{r13,r14,r15}. 
Recent studies have shifted attention from purely asymptotic threshold analysis to finite-length and implementation-oriented designs, such as optimized edge spreading, 6G-oriented rate-compatible constructions, neural or adaptive WD, and multi-dimensional coupling \cite{r16,r17,r18,r19,r20,r21,r22,r23,r24,r25,r26,r27,r28,r29,r30}. 
Meanwhile, progress in quantum LDPC codes and fault-tolerant quantum memories has inspired analogous coupled constructions in the quantum setting, where SC-QLDPC codes have been proposed as quantum counterparts and generalizations of classical multi-dimensional SC codes \cite{r31,r32}.

Although spatial coupling provides excellent asymptotic decoding thresholds, the finite-length performance of SC-LDPC codes is still governed by small harmful substructures in the Tanner graph, especially in the error-floor region. 
For applications requiring extremely low residual error probabilities, it is therefore not enough to know whether a good edge-spreading assignment exists; one also needs to understand how large the feasible design space is. 
In particular, lower-bounding the number of partition matrices that eliminate prescribed short-cycle structures provides a quantitative measure of design feasibility and offers guidance for practical code search. 
This motivates our study of both existence and counting bounds for structure-avoiding  SC-LDPC constructions.

This work complements our recent work \cite{r33}, where a CLLL/MT-type framework provided constructive and counting guarantees for structure-avoiding QC-SC-LDPC designs. 
In contrast, the present paper removes the lifting step and studies only protograph-level edge spreading. 
Using the Combinatorial Nullstellensatz and the Alon--Füredi theorem, we derive sharper nonconstructive existence and counting bounds for eliminating prescribed short-cycle structures. 
For fully connected base graphs, the obtained memory bounds for destroying all $4$-cycles, and all $4$- and $6$-cycles, are asymptotically tight up to a constant factor, although no corresponding construction algorithm is provided. A detailed comparison between the present Nullstellensatz bounds and the CLLL/MT bounds of~\cite{r33}, including the $Z=1$ specialization, is deferred to the same-titled preprint version of this work.
\section{PRELIMINARIES}
\subsection{SC-LDPC Codes}
In this section, we introduce the construction of Type-I SC-LDPC codes from a regular LDPC base matrix. We start with an all-one $(\gamma,\kappa)$ base matrix $\mathbf{H}$, which represents a fully connected bipartite base graph with $\gamma$ check-node types and $\kappa$ variable-node types. 
An SC-LDPC code is obtained by coupling a sequence of identical block-code sections into a spatial chain. 
From the matrix perspective, this coupling operation is commonly referred to as \textit{edge-spreading}. 
Specifically, the nonzero entries of the base matrix $\mathbf{H}$ are divided into $m+1$ component matrices $
\{\mathbf{H}_0,\mathbf{H}_1,\ldots,\mathbf{H}_m\},
$
where each $\mathbf{H}_l$ has the same size as $\mathbf{H}$ and satisfies
\[
\mathbf{H}=\sum_{l=0}^{m}\mathbf{H}_l .
\]
Here, $m$ is called the \textit{memory} of the coupled code.

The edge-spreading pattern is recorded by a $(\gamma,\kappa)$ partition matrix $\mathbf{P}$. 
For $i\in[\gamma]$ and $j\in[\kappa]$, where $[\gamma]=\{0,1,\ldots,\gamma-1\}$ and $[\kappa]=\{0,1,\ldots,\kappa-1\}$, the entry $\mathbf{P}(i,j)=l$ with $0\le l\le m$ indicates that the edge corresponding to $\mathbf{H}(i,j)$ is assigned to the $l$-th component matrix, namely
\[
\mathbf{H}_l(i,j)=1,\qquad 
\mathbf{H}_{l'}(i,j)=0,\quad \forall l'\neq l .
\]
In general, different spatial positions may use different base matrices or different edge-spreading patterns. 
However, to reduce structural complexity and facilitate a regular coupled construction, we assume that all spatial positions share the same base matrix and the same partition matrix throughout this paper.

Given a coupling length $L$, the terminated SC-LDPC parity-check matrix is constructed by arranging the component matrices along a diagonal band as
\[
\mathbf{H}_{\mathrm{SC}}=
\begin{bmatrix}
\mathbf{H}_0      & \mathbf{0}        & \cdots       & \mathbf{0} \\
\mathbf{H}_1      & \mathbf{H}_0      & \ddots       & \vdots \\
\vdots            & \mathbf{H}_1      & \ddots       & \mathbf{0} \\
\mathbf{H}_m      & \vdots            & \ddots       & \mathbf{H}_0 \\
\mathbf{0}        & \mathbf{H}_m      & \ddots       & \mathbf{H}_1 \\
\vdots            & \ddots            & \ddots       & \vdots \\
\mathbf{0}        & \cdots            & \mathbf{0}   & \mathbf{H}_m
\end{bmatrix}.
\]
Equivalently, the block at block-row $t$ and block-column $s$ is given by $\mathbf{H}_{t-s}$ if $0\le t-s\le m$, and by the all-zero matrix otherwise. 
Thus, $\mathbf{H}_{\mathrm{SC}}$ has size $(L+m)\gamma \times L\kappa$. 
Under this notation, the considered Type-I SC-LDPC code is characterized by the parameter tuple
$(\gamma,\kappa,m,L)$
together with the partition matrix $\mathbf{P}$.

\subsection{Major Harmful Structures}

Although SC-LDPC codes improve the iterative-decoding performance of their uncoupled counterparts, they may still exhibit an error floor at high SNRs over the AWGN channel. 
This degradation is mainly caused by harmful small structures in the Tanner graph.
We recall the relevant definitions below.

\begin{lemma}\label{lemma_1}
Let $\mathcal{C}_{2g}$ be the set of cycle-2g candidates in the base matrix, where $g \in \mathbb{N},~g \geq 2$. Denote $c_{2g} \in \mathcal{C}_{2g}$ by $(j_1, i_1, j_2, i_2, \dots, j_g, i_g)$, where $(i_k, j_k)$ and $(i_k, j_{k+1})$ for $1 \leq k \leq g$, with $j_{g+1} = j_1$, are nodes of $c_{2g}$ in $\mathbf{H}$, $\mathbf{P}$. Then $c_{2g}$ becomes a cycle candidate in the protograph after the edge-spreading operation if and only if the following condition holds \cite{r34}:

\begin{align}\label{e_1}
\sum_{k=1}^{g} \mathbf{P}(i_k, j_k) = \sum_{k=1}^{g} \mathbf{P}(i_k, j_{k+1}). 
\end{align}

\end{lemma}

\section{Upper bounds via the Combinatorial Nullstellensatz}

In this section, we establish the main algebraic existence result of this paper. Instead of analyzing the probability that a harmful structure remains active under random edge spreading, we encode the edge-spreading design problem as a polynomial non-vanishing problem over a finite grid. Each edge of the fully connected base graph is associated with a variable, and the activation condition of a cycle-based harmful structure is represented by the vanishing of a polynomial in these variables. Hence, if the product of the corresponding polynomials is nonzero at some point of the grid, then the associated edge-spreading assignment destroys all prescribed harmful structures simultaneously. By applying the Combinatorial Nullstellensatz to this product polynomial, we obtain explicit sufficient memory conditions guaranteeing the existence of such assignments. These conditions provide deterministic upper bounds on the memory required for structure-avoiding SC-LDPC code constructions.

\begin{lemma}[Grid form of the Combinatorial Nullstellensatz {\cite{r35}}]
Let $F$ be a field, and let $S_i\subset F$ be finite nonempty sets.
If $f\in F[x_1,\ldots,x_n]$ is nonzero and satisfies
\[
    \deg_{x_i} f < |S_i|,\qquad i=1,\ldots,n,
\]
then there exists
$
    (s_1,\ldots,s_n)\in S_1\times\cdots\times S_n
$
such that
\[
    f(s_1,\ldots,s_n)\neq 0.
\]
\end{lemma}

\begin{definition}[Cycle hitting set and load]
\label{def:cycle_hitting_set_load}
Let $B$ be a base graph with edge set $E(B)$, and let $\mathcal{H}=\{H_1,\ldots,H_k\}$
be a collection of avoidable harmful structures in $B$. For each
$H_i\in\mathcal{H}$, denote by $N_b^i=\{c_{i,1},c_{i,2},\ldots,c_{i,n_b^i}\}$ the set of fundamental cycles of $H_i$. Define
$
    \mathcal{C}_{\mathcal{H}}
    =
    \bigcup_{i=1}^{k} N_b^i
$
as the set of all fundamental cycles appearing in the structures of
$\mathcal{H}$, where repeated cycles are identified. A subset
$
    \mathcal{T}\subseteq \mathcal{C}_{\mathcal{H}}
$
is called a \textit{cycle hitting set} of $\mathcal{H}$ if
$
    \mathcal{T}\cap N_b^i \neq \emptyset,
    (i=1,2,\ldots,k ).
$
That is, $\mathcal{T}$ contains at least one fundamental cycle from each
harmful structure. For a cycle hitting set $\mathcal{T}$ and an edge $e\in E(B)$, define
the \textit{load} of $e$ induced by $\mathcal{T}$ as
$
    \lambda_{\mathcal{T}}(e)
    =
    \bigl|\{c\in \mathcal{T}: e\in E(c)\}\bigr|.
$
The \textit{maximum load} of $\mathcal{T}$ is defined as
$
    W(\mathcal{T})
    =
    \max_{e\in E(B)} \lambda_{\mathcal{T}}(e).
$
Finally, the \textit{minimum hitting-set load} of $\mathcal{H}$ is defined as
$
    W_{\mathcal{H}}^{\mathrm{hit}}
    =
    \min_{\mathcal{T}} W(\mathcal{T}),
$
where the minimum is taken over all cycle hitting sets
$\mathcal{T}$ of $\mathcal{H}$.
\end{definition}

\begin{theorem}[Cycle-hitting-set/load-bound sufficient condition]
\label{theorem:chs_load_bound}
Consider a SC-LDPC code constructed from a $\gamma \times \kappa$
fully-connected block-code base graph $B$ without parallel edges. Let
$E(B)=\{e_1,\ldots,e_{\gamma\kappa}\}$ be the edge set of $B$, and let $S=\{a_0,a_1,\ldots,a_{m_t}\}$ be the set induced by the coupling pattern
$\bm{a}=(a_0,a_1,\ldots,a_{m_t})$, where
$|S|=m_t+1$ and \(0 = a_0 < a_1 < \cdots < a_{m_t} = m\).

Let $\mathcal{H}=\{H_1,\ldots,H_k\}$
 be a collection of avoidable harmful structures in $B$, and let
$W_{\mathcal{H}}^{\mathrm{hit}}$ be the minimum hitting-set load defined
in Definition~\ref{def:cycle_hitting_set_load}. If
\[
    m \ge m_t \ge W_{\mathcal{H}}^{\mathrm{hit}}
\]
then there exists an edge-spreading assignment under the coupling
pattern $\bm{a}$ that breaks all harmful structures in
$\mathcal{H}$. Here \(m\) denotes the maximum admissible coupling memory.

\end{theorem}

\begin{IEEEproof}
Assign to each edge $e\in E(B)$ a variable $x_e$, where
$x_e\in S$ indicates the component matrix to which $e$ is
assigned after edge spreading.

For each fundamental cycle $c\in \mathcal{C}_{\mathcal{H}}$, let
$l_c(\mathbf{x})$ denote the algebraic cycle polynomial corresponding
to (\ref{e_1}). The cycle $c$ remains active if and only if
\[
    l_c(\mathbf{x})=0.
\]
Since the base graph has no parallel edges, each fundamental cycle
contains distinct edge variables. Hence $l_c(\mathbf{x})$ is a nonzero
linear polynomial over a field containing $S$.

By Definition~\ref{def:cycle_hitting_set_load}, choose a cycle hitting
set $\mathcal{T}^{\star}$ such that
\[
    W(\mathcal{T}^{\star})
    =
    W_{\mathcal{H}}^{\mathrm{hit}}.
\]
Define
\[
    L_{\mathcal{T}^{\star}}(\mathbf{x})
    =
    \prod_{c\in \mathcal{T}^{\star}} l_c(\mathbf{x}).
\]
Since every $l_c(\mathbf{x})$ is nonzero and the underlying coefficient
ring is an integral domain,
\[
    L_{\mathcal{T}^{\star}}(\mathbf{x})\not\equiv 0.
\]

If
\[
    L_{\mathcal{T}^{\star}}(\mathbf{x})\neq 0,
\]
then
\[
    l_c(\mathbf{x})\neq 0,
    \qquad \forall c\in\mathcal{T}^{\star}.
\]
Thus every cycle selected in $\mathcal{T}^{\star}$ is made inactive.
Since $\mathcal{T}^{\star}$ is a cycle hitting set, every harmful
structure $H_i\in\mathcal{H}$ contains at least one selected fundamental
cycle. Therefore, every harmful structure is broken.

It remains to show that such an assignment exists. For any edge
$e\in E(B)$, the variable $x_e$ can appear in
$L_{\mathcal{T}^{\star}}(\mathbf{x})$ only through the factors
$l_c(\mathbf{x})$ with $e\in E(c)$. Since each $l_c(\mathbf{x})$ is
linear in every edge variable, we have
\[
    \deg_{x_e} L_{\mathcal{T}^{\star}}(\mathbf{x})
    \le
    \lambda_{\mathcal{T}^{\star}}(e)
    \le
    W(\mathcal{T}^{\star})
    =
    W_{\mathcal{H}}^{\mathrm{hit}} .
\]
If
\[
    |S|=m_t+1 > W_{\mathcal{H}}^{\mathrm{hit}},
\]
then
\[
    \deg_{x_e} L_{\mathcal{T}^{\star}}(\mathbf{x})
    <
    |S|,
    \qquad \forall e\in E(B).
\]
By the grid form of the Combinatorial Nullstellensatz, the nonzero
polynomial $L_{\mathcal{T}^{\star}}(\mathbf{x})$ cannot vanish on the
entire Cartesian product $S^{|E(B)|}$. Hence there exists
\[
    \mathbf{x}^{\star}\in S^{|E(B)|}
\]
such that
\[
    L_{\mathcal{T}^{\star}}(\mathbf{x}^{\star})\neq 0.
\]
This assignment makes every cycle in $\mathcal{T}^{\star}$ inactive and
therefore breaks every harmful structure in $\mathcal{H}$.

Finally, since $W_{\mathcal{H}}^{\mathrm{hit}}$ is an integer, the
condition
\[
    m_t+1 > W_{\mathcal{H}}^{\mathrm{hit}}
\]
is equivalent to
\[
    m_t \ge W_{\mathcal{H}}^{\mathrm{hit}}.
\]
Together with $m\ge m_t$, this gives the sufficient condition
\[
    m \ge m_t \ge W_{\mathcal{H}}^{\mathrm{hit}}.
\]
\end{IEEEproof}

\begin{corollary}
\label{cor:all_fundamental_cycles}
Consider a SC-LDPC code constructed from a $\gamma\times\kappa$
fully-connected block-code base graph without parallel edges under the
coupling pattern $\bm{a}=(a_0,a_1,\ldots,a_{m_t})$. Let
$\mathcal{H}=\{H_1,\ldots,H_k\}$ be a collection of avoidable harmful
structures, and let $N_b^i$ denote the set of fundamental cycles of
$H_i$. Define
\[
    \mathcal{T}_{\mathcal{H}}
    =
    \bigcup_{i=1}^{k} N_b^i .
\]
For each edge $e$, let
\[
    \lambda_{\mathcal{H}}(e)
    =
    |\{c\in \mathcal{T}_{\mathcal{H}}: e\in E(c)\}|,
    \qquad
    W_{\mathcal{H}}
    =
    \max_{e}\lambda_{\mathcal{H}}(e).
\]
Then all harmful structures in $\mathcal{H}$ can be broken if
\[
    m \ge m_t \ge W_{\mathcal{H}} .
\]
\end{corollary}

\begin{IEEEproof}
Choose the cycle hitting set as all fundamental cycles, i.e.,
$\mathcal{T}=\mathcal{T}_{\mathcal{H}}=\bigcup_{i=1}^{k}N_b^i$.
Clearly, $\mathcal{T}\cap N_b^i\neq\emptyset$ for every
$H_i\in\mathcal{H}$, so $\mathcal{T}$ is a valid cycle hitting set.
Moreover, its maximum edge load is exactly $W_{\mathcal{H}}$. Therefore,
the result follows directly from Theorem~\ref{theorem:chs_load_bound}.
\end{IEEEproof}

\begin{corollary}
    For any SC-LDPC code constructed with respect to a $\gamma \times \kappa$ fully-connected block code base graph, with no parallel edges, a sufficient condition for having a girth of at least 6 is  $m \ge m_t \ge (\gamma-1)(\kappa-1)$.
\end{corollary}

\begin{IEEEproof}
       When $\mathcal{H}=\mathcal{C}_4$, $W_{\mathcal{C}_4}=(\gamma-1)(\kappa-1)$.
\end{IEEEproof}

\begin{remark}[Explicit \(4\)-cycle elimination]
For the fully connected \(\gamma\times\kappa\) base matrix, the memory bound
\(m=(\gamma-1)(\kappa-1)\) can be achieved by a simple explicit assignment.
Index the rows and columns by \(i\in\{0,\ldots,\gamma-1\}\) and
\(j\in\{0,\ldots,\kappa-1\}\), and set
\begin{equation}
    P(i,j)=ij .
\end{equation}
Then \(\max_{i,j}P(i,j)=(\gamma-1)(\kappa-1)\). For any \(4\)-cycle determined
by \(i_1\neq i_2\) and \(j_1\neq j_2\), the difference in the two sides of the
cycle activation condition is
\begin{align}
& P(i_1,j_1)+P(i_2,j_2)-P(i_1,j_2)-P(i_2,j_1) \nonumber\\
&\quad =
(i_1-i_2)(j_1-j_2)\neq 0 .
\end{align}
Therefore no \(4\)-cycle remains active after edge spreading, and the resulting
coupled protograph has girth at least six.
\end{remark}

\begin{corollary}
    For any SC-LDPC code constructed with respect to a $\gamma \times \kappa$ fully-connected block code base graph, with no parallel edges, a sufficient condition for having a girth of at least 8 is  $m \ge m_t \ge (\gamma-1)(\gamma-2)(\kappa-1)(\kappa-2)+(\gamma-1)(\kappa-1)$.
\end{corollary}

\begin{IEEEproof}
       When $\mathcal{H}=\mathcal{C}_4 \cup \mathcal{C}_6$, $W_{\mathcal{H}}=(\gamma-1)(\gamma-2)(\kappa-1)(\kappa-2)+(\gamma-1)(\kappa-1)$.
\end{IEEEproof}

\begin{remark}
In the construction of SC-LDPC codes, previous works have established lower bounds for the destruction of cycles of length 4 and 6. Specifically, for $\gamma=3$, the lower bound for destroying all 4-cycles is given by $\lceil{\frac{\kappa-1}{2}}\rceil$, and the lower bound for destroying both 4-cycles and 6-cycles is given by $\lceil{\frac{\kappa(\kappa-1)}{8}}\rceil$ \cite{r34}. In contrast, this work presents new upper bounds that are asymptotically tight up to a constant. Specifically, the upper bound for destroying all 4-cycles is $2(\kappa-1)$, and for destroying both 4-cycles and 6-cycles, the upper bound is $2(\kappa-1)(\kappa-2) + 2(\kappa-1)$. These results not only improve upon previous upper bounds but also demonstrate that the upper and lower bounds are asymptotically tight, providing a refined and more accurate characterization of the sufficient conditions for the destruction of harmful cycles in the construction of SC-LDPC codes.

\end{remark}

\section{Counting bounds from the existence theorem}
\label{sec:counting_bounds}
This section turns the preceding existence theorem into quantitative
counting bounds: by applying the Alon--Füredi theorem to the elimination
polynomial on the finite edge-spreading grid, we lower-bound the number
of assignments that avoid the prescribed harmful structures.

Let
\[
    A=\prod_{i=1}^n A_i
\]
be a finite grid over a field $F$. For $f\in F[t_1,\ldots,t_n]$, define
\[
    \mathcal U_A(f)=\{x\in A:f(x)\neq 0\}.
\]
We also write

\[
\mathfrak m(a_1,\ldots,a_n;N)
=
\min_{\substack{
\mathbf y\in\mathbb Z^n\\
1\le y_i\le a_i,\ \sum_i y_i=N
}}
\prod_{i=1}^n y_i .
\]

\begin{lemma}[Alon--Füredi theorem {\cite{r36}}]
Let $A=\prod_{i=1}^n A_i\subset F^n$ be a finite grid. If
$f\in F[t_1,\ldots,t_n]$ does not vanish on all points of $A$, then
\[
    \#\mathcal U_A(f)
    \ge
    \mathfrak m
    \left(
    \#A_1,\ldots,\#A_n;
    \sum_{i=1}^n \#A_i-\deg f
    \right).
\]
\end{lemma}

For a uniform grid, this lower bound has the following explicit form.

\begin{lemma}
\label{lem:uniform_AF_bound}
Let \(n\ge 1\), \(s\ge 2\), and \(D\in\mathbb Z\) satisfy
\[
    0\le D\le n(s-1).
\]
Write
\[
    D=q(s-1)+r,
    \qquad
    q\in\mathbb Z_{\ge0},
    \qquad
    0\le r<s-1.
\]
Then
\[
\mathfrak m
\left(
\underbrace{s,\ldots,s}_{n};
ns-D
\right)
=
\begin{cases}
s^{n-q}, & r=0,\\[3pt]
(s-r)s^{n-q-1}, & r>0.
\end{cases}
\]
\end{lemma}

\begin{IEEEproof}
Set $z_i=s-y_i$. Then
\[
0\le z_i\le s-1,
\qquad
\sum_{i=1}^n z_i=D,
\qquad
\prod_{i=1}^n y_i=\prod_{i=1}^n(s-z_i).
\]
For fixed $\sum_i z_i=D$, the product is minimized by concentrating
the deficit $D$ in as few coordinates as possible. Hence, up to
permutation,
\[
z_1=\cdots=z_q=s-1,\qquad z_{q+1}=r,
\]
when $r>0$, and all remaining $z_i$ are zero. This gives the stated
formula.
\end{IEEEproof}

\begin{theorem}[Counting consequence of the existence theorem]
\label{thm:counting_from_existence}
Consider a SC-LDPC code constructed from a base graph $B$ with no
parallel edges. Let $E(B)$ be its edge set and set
$
    n=|E(B)|.
$
Let
$
    S=\{a_0,a_1,\ldots,a_{m_t}\}\subset F, 
    |S|=m_t+1,
$
and let
$
    A=S^{E(B)}
$
be the set of all edge-spreading assignments using the coupling pattern
$\mathbf a=(a_0,a_1,\ldots,a_{m_t})$.

Let
\[
    \mathcal H=\{H_1,\ldots,H_k\}
\]
be a collection of avoidable harmful structures, and let $W_{\mathcal H}$ be the corresponding weight parameter from
the existence theorem. Let $L_{\mathcal H}$ be the elimination
polynomial for $\mathcal H$, so that
\[
    L_{\mathcal H}(x)\neq 0
\]
implies that the assignment $x$ breaks every harmful structure in
$\mathcal H$.

If
\[
    m\ge m_t\ge W_{\mathcal H},
\]
then the existence theorem implies that $L_{\mathcal H}$ does not vanish
identically on $A$. Here \(m\) denotes the maximum admissible coupling memory. Consequently, the number of edge-spreading
assignments that break all structures in $\mathcal H$ is at least
\[
\resizebox{0.98\columnwidth}{!}{$
\#\mathcal U_A(L_{\mathcal H})
\ge
\mathfrak m
\left(
\underbrace{m_t+1,\ldots,m_t+1}_{n};
n(m_t+1)-\deg(L_{\mathcal H})
\right).
$}
\]
\end{theorem}

\begin{IEEEproof}
By the existence theorem, the condition
\[
m\ge m_t\ge W_{\mathcal H}
\]
guarantees that there exists an assignment $x\in A$ such that
\[
L_{\mathcal H}(x)\neq 0.
\]
Hence $L_{\mathcal H}$ does not vanish on all points of $A$. Applying
the Alon--Füredi theorem to $L_{\mathcal H}$ over the grid
$A=S^{E(B)}$ gives
\[
\resizebox{0.98\columnwidth}{!}{$
\#\mathcal U_A(L_{\mathcal H})
\ge
\mathfrak m
\left(
\underbrace{m_t+1,\ldots,m_t+1}_{n};
n(m_t+1)-\deg(L_{\mathcal H})
\right).
$}
\]
Since every point of $\mathcal U_A(L_{\mathcal H})$ breaks all harmful
structures in $\mathcal H$, the same quantity is a lower bound on the
number of valid edge-spreading assignments.
\end{IEEEproof}

\begin{corollary}[Counting assignments eliminating all $4$-cycles]
\label{cor:C4_counting_compact}
Consider a SC-LDPC code constructed from a fully connected
$\gamma\times\kappa$ base graph with no parallel edges. Let
$
    S=\{a_0,a_1,\ldots,a_{m_t}\}$,
    $A=S^{\gamma\kappa},
$
where the elements of $S$ are distinct elements of a field $F$. Let
$\mathcal C_4$ denote the set of unordered $4$-cycle candidates in the
base graph. If
\[
    m\ge m_t\ge (\gamma-1)(\kappa-1),
\]
then the number of edge-spreading assignments in $A$ that eliminate all
$4$-cycles is at least
\[
\#\mathcal U_A(L_{\mathcal C_4})
\ge
\begin{cases}
(m_t+1)^{\gamma\kappa-q}, & r=0,\\[3pt]
(m_t+1-r)(m_t+1)^{\gamma\kappa-q-1}, & 1\le r<m_t,
\end{cases}
\]
where
\[
    D_4=\binom{\gamma}{2}\binom{\kappa}{2},
    \qquad
    D_4=qm_t+r,
    \qquad
    0\le r<m_t.
\]
Every such assignment yields a SC-LDPC construction with girth at
least six.
\end{corollary}

\begin{IEEEproof}
A $4$-cycle candidate is uniquely determined by choosing two distinct
check-node indices and two distinct variable-node indices. Hence
\[
    |\mathcal C_4|
    =
    \binom{\gamma}{2}\binom{\kappa}{2}
    =
    D_4.
\]

For the $4$-cycle determined by rows $i_1,i_2$ and columns $j_1,j_2$,
define
\[
    \ell_{i_1,i_2;j_1,j_2}
    =
    x_{i_1j_1}+x_{i_2j_2}
    -
    x_{i_1j_2}-x_{i_2j_1}.
\]
The condition
\[
    \ell_{i_1,i_2;j_1,j_2}\neq 0
\]
is sufficient to destroy this $4$-cycle after edge spreading. Therefore
we define
\[
    L_{\mathcal C_4}
    =
    \prod_{c_4\in\mathcal C_4}\ell_{c_4}.
\]
Since each factor is linear,
\[
    \deg(L_{\mathcal C_4})
    =
    |\mathcal C_4|
    =
    D_4.
\]

For $4$-cycles, the existence parameter is
\[
    W_{\mathcal C_4}=(\gamma-1)(\kappa-1).
\]
Thus the condition
\[
    m\ge m_t\ge (\gamma-1)(\kappa-1)
\]
allows us to apply Theorem~\ref{thm:counting_from_existence} with
\[
    n=\gamma\kappa,
    \qquad
    s=m_t+1,
    \qquad
    D=D_4.
\]
Using Lemma~\ref{lem:uniform_AF_bound} and the decomposition
\[
    D_4=qm_t+r,
    \qquad
    0\le r<m_t,
\]
gives the stated lower bound. Since every assignment counted by
$\mathcal U_A(L_{\mathcal C_4})$ destroys all $4$-cycle candidates, the
resulting bipartite graph contains no $4$-cycles. As there are no
parallel edges, the corresponding SC-LDPC construction has girth at
least six.
\end{IEEEproof}

\begin{example}[$3\times 5$ fully connected base graph]
\label{ex:3by5_C4_counting}
Consider the fully connected $3\times 5$ base matrix.
Here $\gamma=3$ and $\kappa=5$, so the number of edges in the base graph is
$
n=\gamma\kappa=15.
$
The number of unordered $4$-cycle candidates is
\[
    D_4=\binom{3}{2}\binom{5}{2}
    =3\cdot 10
    =30.
\]

By Corollary~\ref{cor:C4_counting_compact}, a sufficient memory condition
for eliminating all $4$-cycles is
\[
    m\ge m_t\ge (\gamma-1)(\kappa-1)
    =2\cdot 4
    =8.
\]
Take
\[
    S=\{0,1,\ldots,8\},
    \qquad |S|=m_t+1=9.
\]
Since
\[
    D_4=30=3\cdot 8+6,
\]
we have
\[
    q=3,
    \qquad
    r=6.
\]
Therefore, the counting lower bound gives
\[
\#\mathcal U_A(L_{\mathcal C_4})
\ge
(m_t+1-r)(m_t+1)^{\gamma\kappa-q-1}
=
(9-6)9^{15-3-1}.
\]
Equivalently,
\[
\#\mathcal U_A(L_{\mathcal C_4})
\ge
3\cdot 9^{11}
=
94,143,178,827.
\]
Since the total number of edge-spreading assignments over $S$ is
\[
    |A|=9^{15}
    =
    205,891,132,094,649,
\]
the lower bound also implies that a uniformly random assignment over
$S^{15}$ eliminates all $4$-cycles with probability at least
\[
    \frac{3\cdot 9^{11}}{9^{15}}
    =
    \frac{1}{2187}.
\]

\end{example}

\medskip
\noindent\textbf{Comparison with the CLLL/MT bounds at $Z=1$.}
It is also useful to compare the present protograph-level
bounds with the CLLL/MT bounds in~\cite{r33}
after specializing the lifting degree to $Z=1$. In this case,
the lifting variables carry no additional freedom, so the
comparison reduces to a pure edge-spreading problem. For
4-cycle elimination and fixed $\gamma$, the CLLL condition
requires
\[
    m_{\rm CLLL} \sim C_{\rm CLLL}(\gamma)\kappa,
    \qquad
    \frac{C_{\rm CLLL}(\gamma)}{\gamma-1}
    =
    \begin{cases}
    2e, & \gamma=3,\\
    20e/9, & \gamma=4,\\
    512/81, & \gamma\ge 5,
    \end{cases}
\]
whereas Corollary~2 gives the smaller deterministic condition
$m\ge(\gamma-1)(\kappa-1)$. The MT-output condition is also
larger, with
\[
    m_{\rm MT}\sim \frac{4e(2\gamma-3)}{3}\kappa .
\]
For instance, for the $3\times5$ fully connected base graph,
the present bound gives $m=8$, while the CLLL and MT
specializations require $m=35$ and $m=37$, respectively.
The counting bounds show the same behavior for feasible
edge-spreading assignments: at $m=8$, Corollary~4 gives
\[
    \#U_A(L_{C_4})\ge 3\cdot 9^{11}
    \approx 9.41\times 10^{10},
\]
whereas the CLLL feasible-solution bound at its minimum
admissible memory $m=35$ gives
\[
    36^{15}\left(1-\frac{2}{20}\right)^{300}
    \approx 4.14\times 10^9 .
\]
At the same memory $m=35$, the present Alon--F\"uredi
bound becomes
\[
    \#U_A(L_{C_4})\ge 6\cdot 36^{14}
    \approx 3.68\times 10^{22}.
\]
The MT-output diversity bound is of a different algorithmic
nature; at $m=37$ it gives
\[
    \frac{38^{15}}{\exp(3/2)}
    \approx 1.11\times 10^{23},
\]
which is comparable to the Alon--F\"uredi value
$8\cdot38^{14}\approx 1.05\times10^{23}$ at the same memory,
but it requires the larger MT memory condition.
\section{Conclusion}

This paper developed an algebraic-combinatorial framework for
protograph-level edge-spreading designs of high-memory SC-LDPC
codes. By representing cycle activation conditions as polynomial
vanishing constraints over finite grids, we obtained sufficient
memory conditions for eliminating prescribed cycle-induced harmful
structures. For fully connected $(\gamma,\kappa)$ base graphs, the
bounds explicitly characterize the memory needed to destroy all
$4$-cycles and all $4$- and $6$-cycles, and are asymptotically tight
up to a constant factor for fixed $\gamma$. We also used the
Alon--F{\"u}redi theorem to lower-bound the number of feasible
edge-spreading assignments, giving a quantitative measure of the
available design space. The framework is nonconstructive; efficient
algorithms approaching these bounds and extensions to larger
trapping-set or absorbing-set families remain important directions
for future work.


\bibliographystyle{IEEEtran}

\bibliography{IEEEabrv,ref}

\end{document}